\documentclass[a4paper,12pt,reqno]{article} 

\usepackage{fourier} 
\usepackage[scaled=0.875]{helvet} 

\usepackage{titlesec}
\titlespacing\section{0pt}{30pt plus 4pt minus 2pt}{5pt plus 2pt minus 2pt}
\titlespacing\subsection{0pt}{15pt plus 4pt minus 2pt}{5pt plus 2pt minus 2pt}
\titlespacing\subsubsection{0pt}{12pt plus 4pt minus 2pt}{0pt plus 2pt minus 2pt}

\usepackage{float} 

\usepackage{tabularx}
\newcolumntype{L}[1]{>{\raggedright\arraybackslash}p{#1}}
\newcolumntype{C}[1]{>{\centering\arraybackslash}p{#1}}
\newcolumntype{R}[1]{>{\raggedleft\arraybackslash}p{#1}}

\usepackage{tikz}

\usepackage{setspace}                   
\usepackage{enumitem}
\setlist[itemize]{leftmargin=*} 

\usepackage{amsmath,amssymb,amsthm}
\usepackage{wasysym}
\usepackage{indentfirst}
\usepackage{graphicx}
\usepackage[header,title,titletoc]{appendix}
\usepackage{comment}
\usepackage{natbib}             

\usepackage{geometry}
\usepackage[bookmarks,colorlinks]{hyperref} 
\hypersetup{colorlinks,%
           citecolor=blue,%
           filecolor=black,%
           linkcolor=blue,%
           urlcolor=black}
\usepackage[stable]{footmisc} 

\renewcommand{\thefootnote}{\fnsymbol{footnote}}

\newcommand{\bd}[1]{{\bf #1}}
\newcommand{\ft}[1]{\footnote{#1}}

\usepackage{soul} 

\newcommand{\email}[1]{\href{mailto:#1}{\nolinkurl{#1}}}

\newtheorem{defi}{Definition}
\newtheorem{prop}{Proposition}
\newtheorem{lem}{Lemma}
\newtheorem{cor}{Corollary}

\newtheorem{ex}{Example}

\theoremstyle{definition}

\begin{document}

\title{On the Asymptotic Performance of Affirmative Actions in School Choice}

\author{Di Feng\footnotemark[1] \; and \; Yun Liu\footnotemark[2]}

\date{December 10, 2022} 
 
\footnotetext[1]{Faculty of Business and Economics, University of Lausanne, 1015 Lausanne, Switzerland. Email: di.feng@unil.ch.}

\footnotetext[2]{Corresponding author. Center for Economic Research, Shandong University, Jinan 250100, China. Email: yunliucer@sdu.edu.cn.}

\maketitle
\setcounter{footnote}{0}
\renewcommand\thefootnote{\arabic{footnote}}

\thispagestyle{empty}

\begin{abstract} 

This paper analyzes the asymptotic performance of two popular affirmative action policies, \emph{majority quota} and \emph{minority reserve}, under the immediate acceptance mechanism (IAM) and the top trading cycles mechanism (TTCM) in the contest of school choice. The matching outcomes of these two affirmative actions are asymptotically equivalent under the IAM when all students are sincere. Given the possible preference manipulations under the IAM, we characterize the asymptotically equivalent sets of Nash equilibrium outcomes of the IAM with these two affirmative actions. However, these two affirmative actions induce different matching outcomes under the TTCM with non-negligible probability even in large markets.


\medskip

\noindent\bd{JEL Classification Number:} C78; D78; I20

\medskip

\noindent\bd{Keywords:} school choice; immediate acceptance mechanism; top trading cycles mechanism; affirmative actions; large market

\end{abstract}

\maketitle
\thispagestyle{empty}

\section{Introduction}

Affirmative action policies, albeit controversial, attempt giving disadvantaged social groups preferential treatments to improve their socioeconomic status and representation in our societies. In the context of public school choice, many school districts in the United States and around the world often impose affirmative action policies to maintain the racial, ethnic and socioeconomic diversity at schools. The quota-based affirmative action (\emph{majority quota}, henceforth) and the reserve-based affirmative action (\emph{minority reserve}, henceforth) are two popular policy designs in practice. \cite{AS03} formalize the majority quota policy in school choice, which sets a maximum number less than the school's capacity to  \emph{majority} students and leaves the difference to \emph{minority} students.\ft{We term the intended beneficiaries from affirmative action policies as minority students, and all the rest students as majority students; in other words, the distinction between the majority and the minority students does not depend on race, ethnicity, or other single socioeconomic status.} The minority reserve policy proposed by \cite{HYY13}, on the other hand, gives higher priority to minority students up to the point that all reserved seats have been assigned to minorities.

The \emph{immediate acceptance mechanism} (IAM, henceforth; also known as the Boston mechanism) and the \emph{top trading cycles mechanism} (TTCM, henceforth) are two common matching mechanisms for school choice, which are currently being used in many school districts in the U.S. (e.g., Boston, Columbus, Minneapolis, and Seattle), and other OECD countries \citep{M12,CG18}. This paper extends the performance comparison of the majority quota policy and its minority reserve counterpart under the IAM and the TTCM to large school choice markets (i.e., a sequence of random matching markets of different sizes).

We first analyze the asymptotic performance of the majority quota and its minority reserve counterpart under the IAM. Based on the observation that it is very unlikely for any two different students to list the same school with nonzero reserved seats in their preference orders when the market is sufficiently large, our Proposition \ref{prop:iamtrue} implies that the matching outcomes of the IAM with these two competing affirmative action policies are asymptotically equivalent when all students are sincere. However, since students have incentives to manipulate the IAM through strategically reporting their preference orders, the majority quota and its minority reserve counterpart may no longer generate an identical matching outcome in large markets, as students could have distinct incentives to manipulate their preferences under these two affirmative actions. Proposition \ref{prop:iamnash} shows that although the IAM is open to preference manipulations, these two affirmative actions are most likely to produce the same set of Nash equilibrium outcomes under the IAM when the market becomes sufficiently large.

Proposition \ref{prop:ttc} presents our main argument on the asymptotic performance of the majority quota and its minority reserve counterpart under the TTCM in large school choice markets. It indicates that these two affirmative actions produce different matching outcomes under the TTCM with non-negligible probability, even if the number of reserved seats grows at a slower rate of $O(n^a)$ in a sequence of random markets, where $0 \le a < 1/2$ and $n$ is the number of schools in a given random market (see Definition \ref{def:regular}). As the purpose of imposing affirmative actions in school choice is to improve the welfare of minority students, the outcome non-equivalence between these two affirmative actions under the TTCM certainly results in an ambiguous Pareto dominance relationship for minorities in large matching markets (Corollary \ref{cor:welfare}). In addition, given the possibly substantial political, administrative and cognitive costs of transferring from one affirmative action policy to the other, an immediate policy implication of our results is that the IAM is more cost-effective compared to the TTCM in large school choice markets with affirmative actions, in the sense that it is unnecessary to identify the different welfare effects of these two affirmative actions under the IAM when all students are playing their Nash equilibrium strategies. 

The literature on large matching markets has been growing rapidly in recent years. Most studies nevertheless indicate that many existing impossibility results, ranging from incentives to existence and efficiency in finite matching markets, disappear if we admit an asymptotic variant of these properties in large market environments. To our knowledge, only a few studies have demonstrated that the large market approach does not eliminate all the distinct properties of different matching mechanisms.\ft{\cite{KP09} present an example to illustrate that students still have incentives to manipulate the IAM in large school choice markets. \cite{HKN16} show that neither the TTCM nor the IAM asymptotically respect improvements of school quality (i.e., a school matches with a set of more desirable students if it becomes more preferred by students). \cite{CT19} indicate that the inefficiency of the \emph{student optimal stable mechanism} and instability of the TTCM remain significant in large markets when agents have correlated preferences.} Given the fading of many impossibility results in large markets, some researchers have criticized using the large market approach in matching and market design problems, in the sense that the asymptotic analysis of matching mechanisms may be too ``permissive'' to make market design irrelevant.\ft{See the discussions in Section 3.4 of \cite{K15}.} The current paper thus supports the validity of the large market analytic approach, as it still enables us to capture the subtle difference between the matching mechanisms that can \emph{asymptotically} satisfy some desirable proprieties from those that cannot.

\section{Model}                     %
\label{sec:model}                   %

\subsection{School Choice with Affirmative Actions}
\label{ssec:predef}

Let $S$ and $C$ be two finite sets of students and schools, $|S| \ge 2$. There are two types of students, \emph{majority} and \emph{minority}. $S$ is partitioned into two subsets of students based on their types. Denote $S^M$ the set of majority students, and $S^m$ the set of minority students, $S = S^M \cup S^m$ and $S^M \cap S^m = \emptyset$. Each student $s \in S$ has a strict \emph{preference} order $P_s$ over the set of schools and being unmatched (denoted  by $s$). All students prefer to be matched with some school instead of herself, $c \,P_s\, s$, for all $s \in S$. Each school $c \in C$ has a total capacity of $q_c$ seats, $q_c \ge 1$, and a strict \emph{priority} order $\succ$ over the set of students which is complete, transitive, and antisymmetric. Student $s$ is unacceptable by a school if $e \succ_c s$, where $e$ represents an empty seat in school $c$.

For each school $c$ with \emph{majority quota} affirmative action policy, it cannot admit more majority students than its type-specific majority quota $q_c^M \le q_c$, for all $c \in C$. Accordingly, the \emph{minority reserve} policy gives priority to the minority applicants of school $c$ up to its \emph{minority reserve} $r_c^m \le q_c$, $\forall c \in C$, and allows $c$ to accept majority students up to its capacity $q_c$ if there are not enough minority applicants to fill the reserves.

A school choice \emph{market} with affirmative actions is a tuple $\Gamma = (S, C, P, \succ, (q^M, r^m))$, where $P = (P_s)_{s \in S}$, $\succ = (\succ)_{c \in C}$. When comparing the effects of a majority quota policy with its minority reserve counterpart in a market $\Gamma$, we assume $\Gamma$ is either with only majority quota or with only minority reserve, such that $r_c^m + q_c^M =q_c$, $\forall c \in C$, with $q^M = (q_c^M)_{c \in C}$, $r^m = (r_c^m)_{c \in C}$, and $q = (q_c)_{c \in C}$. 

A \emph{matching} $\mu$ is a mapping from $S \cup C$ to the subsets of $S \cup C$  in market $\Gamma$ such that, for all $s \in S$ and $c \in C$: (i) $\mu(s) \in C \cup \{ s\}$; (ii) $\mu(s) = c$ if and only if $s \in \mu(c)$; (iii) $\mu(c) \subseteq S$ and $|\mu(c)| \le q_c$; (iv) $|\mu(c) \cap S^M| \le q_c^M$. That is, a matching specifies the school where each student is assigned to or matched with herself, and the set of students assigned to each school; no school admits more students than its capacity, and no school admits more majority students than its majority quota. 

A matching $\mu$ is \emph{blocked} by a pair of student $s$ and school $c$ with majority quota, if $c P_s \mu(s)$ and either $|\mu(c)| < q_c$ and $s$ is acceptable to $c$, or:
\begin{enumerate}

\item[(i)] $s \in S^m$, $s \succ_c s'$, for some $s' \in \mu(c)$;

\item[(ii)] $s \in S^M$ and $|\mu(c) \cap S^M| < q_c^M$, $s \succ_c s'$, for some $s' \in \mu(c)$;

\item[(iii)] $s \in S^M$ and $|\mu(c) \cap S^M| = q_c^M$, $s \succ_c s'$, for some $s' \in \mu(c) \cap S^M$. 
\end{enumerate}

A matching $\mu$ is \emph{Q-stable}, if $\mu(s) \,P_s\, s$ for all $s \in S$, and has no blocking pair in $\Gamma$ with majority quota.

Accordingly, a matching $\mu$ is \emph{blocked} by a pair of student $s$ and school $c$ with minority reserve, if $c P_s \mu(s)$ and either $|\mu(c)| < q_c$ and $s$ is acceptable to $c$, or:
\begin{enumerate}

\item[(i)] $s \in S^m$, $s \succ_c s'$, for some $s' \in \mu(c)$;

\item[(ii)] $s \in S^M$ and $|\mu(c) \cap S^m| > r_c^m$, $s \succ_c s'$, for some $s' \in \mu(c)$;

\item[(iii)] $s \in S^M$ and $|\mu(c) \cap S^m| \le r_c^m$, $s \succ_c s'$, for some $s' \in \mu(c) \cap S^M$. 
\end{enumerate}

A matching $\mu$ is \emph{R-stable}, if $\mu(s) \,P_s\, s$ for all $s \in S$, and has no blocking pair in $\Gamma$ with minority reserve. 

As the purpose of imposing affirmative actions in school choice markets is to improve the matching outcomes (i.e., welfare) of minority students, we need some \emph{type-specific} criteria to evaluate the welfare effects of affirmative actions on minority students. Given two matchings $\mu$ and $\mu'$, $\mu$ \emph{Pareto dominates $\mu'$ for minorities} if (i) $\mu(s) P_s \mu'(s)$ for at least one $s \in S^m$, and (ii) $\mu(s) R_s \mu'(s)$ for all $s \in S^m$, where $R_s$ represents two matched outcomes that are equally good for $s$.

A matching \emph{mechanism} $f$ is a function that produces a matching $f(\Gamma)$ for each market $\Gamma$. A mechanism $f$ is \emph{strategy-proof} if for each student $s \in S$ and for any $P$, there exists no $P'_s$ such that $\mu(P'_s, P_{-s}) P_s \mu(P)$, where $P_{-s}=(P_i)_{i\in S\backslash s}$; that is, if a mechanism is strategy-proof, each student finds it optimal to report her preferences truthfully regardless of the preferences of other students. Finally, given two mechanisms $f$ and $f'$, we say \emph{$f$ Pareto dominates $f'$ for minorities} if for all $\Gamma$, either $f'(\Gamma) = f(\Gamma)$ for all minorities or $f'(\Gamma)$ Pareto dominates $f(\Gamma)$ for minorities.

\subsection{Large Markets}        %
\label{ssec:large}                     %

A \emph{random market} is a tuple $\tilde{\Gamma} = ((S^M, S^m), C, \succ, (q^M, r^m), k, (\mathcal{A},\mathcal{B}))$,  where $k$ is a positive integer, $\mathcal{A} = (\alpha_c)_{c \in C}$ and $\mathcal{B} = (\beta_c)_{c \in C}$ are the respective probability distributions on $C$, with $\alpha_c, \beta_c > 0$ for each $c \in C$. We assume that $\mathcal{A}$ for majorities to be different from $\mathcal{B}$ for minorities to reflect their distinct favors for schools. 

A sequence of \emph{random markets} is denoted by $(\tilde{\Gamma}^1,\tilde{\Gamma}^2,\dots)$, where $\tilde{\Gamma}^n = ((S^{M,n}, S^{m,n}), C^n, \succ_n, (q^{M,n}, r^{m,n}), k^n, (\mathcal{A}^n, \mathcal{B}^n))$ is a random market of size $n$, with $|C^n| = n$ as the number of schools, $|r^{m,n}|$ the number of seats reserved for minorities, and $|S^n| = |S^{M,n}| + |S^{m,n}|$ as the number of students in market $\tilde{\Gamma}^n$.

Each random market induces a market by randomly generated preference orders of each student $s$ according to the following procedure introduced by \cite{IM05}:
\begin{itemize}
\item[] \textbf{Step 1:} Select a school independently from the distribution $\mathcal{A}$ (resp. $\mathcal{B}$). List this school as the top ranked school of a majority student $s \in S^M$ (resp. minority student $s \in S^m$).


\item[] \textbf{Step $\mathbf{l \le k}$:}  Select a school independently from $\mathcal{A}$ (resp. $\mathcal{B}$) which has not been drawn from steps 1 to step $l-1$. List this school as the $l^{th}$ most preferred school of a majority student $s \in S^M$ (resp. minority student $s \in S^m$).

\end{itemize}

Each majority (resp. minority) student finds these $k$ schools acceptable, and only lists these $k$ schools in her preference order. Let $\tilde{P}_s^n$ be the truthful preference order of student $s$ generated according to the preceding procedure, and  $\tilde{P}^n= (\tilde{P}_s^n)_{s \in S^n}$ be the profile of truthful preference orders. We introduce the following regularity conditions to guarantee the convergence of the random markets sequence.

\begin{defi} \label{def:regular}

Consider majority quotas $q^{M.n}$ and minority reserves $r^{m,n}$ such that $r^{m,n} + q^{M.n} =q^n$. A sequence of random markets $(\tilde{\Gamma}^{1},\tilde{\Gamma}^{2},\dots)$ is \emph{regular}, if there exist $a \in [0, \frac{1}{2})$, $\lambda, \kappa, \theta > 0$, $r \ge 1$, and positive integers $k$ and $\bar{q}$, such that for all $n$:
\begin{enumerate}[leftmargin=*]
\item[(1)] $k^n \le k$;
\item[(2)] $q_c^n \le \bar{q}$ for all $c \in C^n$;
\item[(3)] $|S^n| \le \lambda n$, $\sum_{c \in C} q_c - |S^n| \ge \kappa n$;
\item[(4)] $|r^{m,n}| \le \theta n^a$;
\item[(5)] $\frac{\alpha_c}{\alpha_{c'}} \in [\frac{1}{r}, r]$, $\frac{\beta_c}{\beta_{c'}} \in [\frac{1}{r}, r]$, for all $c, c' \in C^n$;
\item[(6)] $\alpha_c=0$, for all $c \in C^n$ with $q_c^{M,n} = 0$.
\end{enumerate}
\end{defi}

Condition (1) and (2) assume that the length of students' preference orders and the capacity of each school are bounded across schools and markets. Condition (3) requires that the number of students does not grow much faster than the number of schools, while there is an excess supply of school capacities to accommodate all students.\ft{Note that we do not distinguish the growth rate between majority and minority students, as minority students are generically treated as the intended beneficial student groups from affirmative action policies rather than race or other single socioeconomic status; in other words, the number of minority students is not necessarily less than majorities.} Condition (4) requires that the number of seats reserved for minority students grows at a slower rate of $O(n^a)$, where $a \in [0, \frac{1}{2})$. Condition (5) requires that the popularity of different schools, as measured by the probability of being selected by students from $\mathcal{A}$ for majorities and $\mathcal{B}$ for minorities, does not vary too much. Condition (6) requires that a majority student will not select a school that can only accept minority students (i.e., with majority quota $q_c^{M,n} = 0$), as these two affirmative actions trivially induce disparate matching outcomes in any arbitrarily large markets when a majority student applies to a school with zero majority quota.\ft{Most of these regularity conditions have been introduced in the literature of large matching markets; see \cite{KPR13}, \cite{HKN16}, and \cite{Liu22da} for more detailed illustrations.}

We formally define the \emph{asymptotic} outcome equivalence condition of these two affirmative actions in a sequence of random markets of different sizes as follows.

\begin{defi} \label{def:large}
For any random market $\tilde{\Gamma}$, let $\eta_c(\tilde{\Gamma};f,f')$ be probability that school $c \in C^n$ matched with different sets of students which induces $f(\tilde{\Gamma}) \ne f'(\tilde{\Gamma})$. We say two matching mechanisms are \emph{outcome equivalent in large markets}, if for any sequence of random markets $(\tilde{\Gamma}^1,\tilde{\Gamma}^2,\dots)$ that is regular, $\max_{c \in C^n} \eta_c(\tilde{\Gamma}^n; f,f') \to 0$, as $n \to \infty$; that is, for any $\varepsilon > 0$, there exists an integer $m$ such that for any random market $\tilde{\Gamma}^n$ in the sequence with $n>m$ and any $c \in C^n$, we have $\max_{c \in C^n} \eta_c(\tilde{\Gamma}^n;f,f') < \varepsilon$.

\end{defi}

\section{Asymptotic Equilibrium Outcomes Equivalence under the IAM}    
\label{ssec:nash}                          

For each market $\Gamma =(S, C, P, \succ, (q^M, r^m))$, \cite{AS16} adapt the immediate acceptance mechanism (IAM) to school choice with affirmative actions. The \emph{immediate acceptance mechanism with affirmative actions} algorithm runs as follows: 
\begin{itemize}

\item[] \textbf{Step 1}: Each student applies to her most preferred acceptable school (call it school $c$). The school $c$ first considers minority applicants and permanently accepts them up to its minority reserve $r_c^m$ one at a time following its priority order, if $r_c^m > 0$. School $c$ then considers all the applicants who are yet to be accepted, and one at a time following its priority order. It permanently accepts as many students as up to the remaining total capacity while not admitting more majority students than $q_c^M$. The rest (if any) are rejected. 


\item[] \textbf{Step l}: Each student $s$ who was rejected at Step $(l-1)$ applies to her next preferred acceptable choice (call it school $c$, if any). If school $c$ still has an available seat, it first considers minority applicants and permanently accepts them up to its remaining minority reserve one at a time following its priority order, if $r_c^m > 0$. School $c$ then considers all the applicants who are yet to be accepted, and one at a time following its priority order. It permanently accepts as many students as up to the remaining total capacity while not admitting more majority students than its remaining majority quota. The rest (if any) are rejected.

\end{itemize}

The algorithm terminates either when every student is matched to a school or every unmatched student has been rejected by all acceptable schools, which always terminates in a finite number of steps. For a market $\Gamma$, if $r_c^m = 0$, $\forall c \in C$, i.e., a market with only majority quota, the above algorithm reduces to the \emph{IAM with majority quota} (IAM-Q, henceforth). Also, if $q_c^M = q_c$, $\forall c \in C$, i.e., a market with only minority reserve, the above algorithm reduces to the \emph{IAM with minority reserve} (IAM-R, henceforth). 



Under the same preference generation procedure and the regularity conditions defined in Section \ref{ssec:large}, \cite{Liu22da} obverses that the two affirmative actions will result in different matching outcomes under the \emph{student optimal stable mechanism} (SOSM, henceforth) only when some schools have excessive majority applicants and insufficient number of minority applicants; and shows that it is very unlikely for any two different students to list the same school with nonzero reserved seats in their preference orders with either of these two affirmative actions when the market contains sufficiently many schools. Since the underlying matching mechanisms will not affect how students form their preference orders, while students are permanently accepted in each step of the IAM, \cite{Liu22da}'s analysis of the asymptotic performance of these two affirmative actions under the SOSM is clearly valid under the IAM when all students truthfully report their preference.

\begin{prop} \label{prop:iamtrue}
The IAM-Q and its corresponding IAM-R are outcome equivalence in large markets when all students are sincere.
\end{prop}

\begin{proof}

The result follows \cite{Liu22da}'s argument that the probability for any two different students (either majority or minority) to list the same school $c \in C^n$ with nonzero reserved seats in their preference orders converges to zero, when $n \to \infty$. See Expression (A.3) and the succeeding arguments in his Proof of Proposition 2 for details.
\end{proof}


However, it is well-known that the IAM is open to preference manipulations. Once students become strategic in their preference submissions, the majority quota and its minority reserve counterpart may no longer generate an identical matching outcome, because students may have distinct incentives to misreport their preferences under the IAM-Q and its IAM-R counterpart (see Example \ref{ex1} in Appendix \ref{app:example}).

To analyze the asymptotic performance of these two affirmative actions under the manipulable IAM, we first define the IAM-Q and its corresponding IAM-R as a \emph{preference revelation game} in a sequence of random markets. Formally, given a regular random market $\tilde{\Gamma}^n$ of size $n$, $n=1,2,\dots$, a mechanism $f$ and students' corresponding truthful preference profile $\tilde{P}^n=(\tilde{P}^n_s)_{s\in S^n}$, denote $G_f(\tilde{\Gamma}^n)=(\mathcal{P}^n , \tilde{P}^n, f)$ the preference revelation game induced by $f$, where $\mathcal{P}^n$ is the \emph{strategy space} of each student (i.e., all the possible stated preferences over schools), and $f$ is the \emph{outcome function} in which each student evaluates her assignments according to $\tilde{P}^n$.

\begin{defi} \label{def:nash}
A strategy profile $P^* \in \Pi_{s\in S^n} \mathcal{P}^n$ is a \emph{Nash equilibrium} of $G_f(\tilde{\Gamma}^n)$, if for each $s\in S^n$, there is no strategy $P'_s \in \mathcal{P}^n$ such that $ f_s( P'_s,P^*_{-s} ) \mathbin{\tilde{P}^n_s} f_s(P^*)$, where $P^*_{-s}=(P_i^*)_{i\in S^n\backslash s}$. Also, given a Nash equilibrium $P^*$, its corresponding Nash equilibrium outcome is $f(P^*)$.
\end{defi}

For any regular random market $\tilde{\Gamma}^n$ in the sequence of markets $(\tilde{\Gamma}^1,\tilde{\Gamma}^2,\dots)$, let $\tilde{P}_s^n$ be the truthful preference order of student $s$ generated according to the procedure defined in Section \ref{ssec:large}, and $\xi^q(\tilde{P}^n)$ (resp. $\xi^r(\tilde{P}^n)$) be the set of Q-stable (resp. R-stable) matchings under the profile of truthful preferences $\tilde{P}^n$ with majority quota (resp. minority reserve), where $\tilde{P}^n= (\tilde{P}_s^n)_{s \in S^n}$. 


\begin{lem} 
The probability that $\xi^q(\tilde{P}^n) = \xi^r(\tilde{P}^n)$ converges to one, as $n \to \infty$.

\end{lem}

\begin{proof}
(i) $\xi^r(\tilde{P}^n) \subseteq \xi^q(\tilde{P}^n)$. ~Given a regular random market $\tilde{\Gamma}^n$, let $\mu \notin \xi^q(\tilde{P}^n)$ be a non R-stable matching in it, i.e., $\mu$ is blocked by a pair of student and school $(s,c) \in (S^n,C^n)$ when $\tilde{\Gamma}^n$ has the majority quota $q^{M,n}$. By the definition of R-stability, $(s,c)$ also blocks $\mu$ in $\tilde{\Gamma}^n$ with the corresponding minority reserve $r^{m,n} = q^n- q^{M,n}$. Thus, we have $\mu \notin \xi^r(\tilde{P}^n)$ in $\tilde{\Gamma}^n$.

\smallskip

(ii) $\xi^q(\tilde{P}^n) \subseteq \xi^r(\tilde{P}^n)$. ~Let $\mu \in \xi^q(\tilde{P}^n)$ be a Q-stable matching in a given regular random market $\tilde{\Gamma}^n$ of size $n$. We demonstrate that $\mu$ is \emph{asymptotically} R-stable when the size of the market grows up; that is, the probability that $\mu \in \xi^r(\tilde{P}^n)$ converges to one in the sequence of markets $(\tilde{\Gamma}^1,\tilde{\Gamma}^2,\dots)$, as $n \to \infty$. 

As a majority student will never list a school $c$ with $q_c^M = 0$ (Condition (6) of Definition \ref{def:regular}), a Q-stable matching $\mu$ in $\tilde{\Gamma}^n$ with majority quota $q^{M,n}$ is blocked by a pair of $(s,c)$ in $\tilde{\Gamma}^n$ with the corresponding minority reserve $r^{m,n}$, can only occur when $c \,\tilde{P}_s^n\, \mu(s)$, $|\mu(c) \cap S^M| = q_c^M$, $|\mu(c)| < q_c$, and $s' \succ_c s$, for all $s' \in \mu(c) \cap S^{M,n}$ and $s \in S^{M,n} \backslash \mu(c)$; that is, school $c$ has excessive majority applicants and insufficient number of minority applicants in $\tilde{\Gamma}^n$. Recall Proposition \ref{prop:iamtrue}, we know that it is very unlikely for any two distinct students (i.e., $s$ and $s'$ here) to list the same school $c$ with nonzero reserved seats in their preference orders when $n$ becomes sufficiently large. This implies that the probability for any pair $(s,c) \in (S^n,C^n)$ forming a blocking pair in $\tilde{\Gamma}^n$ with $r^{m,n}$ but not in $\tilde{\Gamma}^n$ with the corresponding $q^{M,n}$ converges to zero, as $n \to \infty$. Thus, we have the probability that $\mu \in \xi^r(\tilde{P}^n)$ converges to one, as $n \to \infty$.
\end{proof}

\begin{lem} \label{lem3}

(1) The set of Nash equilibrium outcomes of the IAM-Q is equal to $\xi^q(\tilde{P}^n)$in each $\tilde{\Gamma}^n$, $n=1,2,\dots$.

\noindent (2) The set of Nash equilibrium outcomes of the IAM-R is equal to $\xi^r(\tilde{P}^n)$ in each $\tilde{\Gamma}^n$, $n=1,2,\dots$.

\end{lem}

\begin{proof}
(1) \,The market-wise equivalence between the set of Nash equilibrium outcomes of the IAM-Q and the set of Q-stable matchings under the truthful preferences in each random market of size $n$, has been given by Theorem 3 of \cite{ES06}. Thus, we only need to prove the second part.

(2.i) \,Given a regular random market $\tilde{\Gamma}^n$ of size $n$, $n=1,2,\dots$, and the corresponding preference revelation game of IAM-R, let $P'$ be an arbitrary strategy profile and matching $\mu$ be its associated outcome. Suppose that $\mu$ is not R-stable under the truthful preference profile $\tilde{P}^n$, we can thus find a pair of student and school $(s,c) \in (S^n,c^n)$ such that $c \tilde{P}^n_s \mu(s)$, and either $s \succ_c s'$ for some $s' \in \mu(c)$, or $|\mu(c)| < q_c$ and $s$ is acceptable to $c$. This implies that $c$ is not at the top in $P'_s$, because otherwise student $s$ would have been assigned to school $c$. Let $P''_s$ be an alternative preference order of $s$ in which $c$ is positioned as her first choice. Clearly, $s$ will be assigned to $c$ under the strategy profile $(P''_s, P'_{-s})$, where $P'_{-s} = (P'_{i})_{i \in S^n \backslash s}$. Thus, $P'$ is not a Nash equilibrium, as $P''_s$ offers a profitable deviation for student $s$ at $P'$ given that $c \tilde{P}^n_s \mu(s)$. Also, since $P'$ is arbitrarily chosen, the non R-stable matching $\mu$ cannot be obtained in the set of Nash equilibrium outcomes.

(2.ii) \,Let $\mu$ be a R-stable matching under $\tilde{P}^n$ in the regular random market $\tilde{\Gamma}^n$. We show that there exists a Nash equilibrium $P^*$, such that its associated outcome is $\mu$. For each student $s\in S^n$, let $P^*_s$ be the preference order of student $s$ such that school $\mu(s)$ is positioned at the top, i.e., $\mu(s) P^*_s c'$, $\forall c'\in C^n \backslash \mu(s)$. Thus, at $P^*$ the IAM-R will terminate at Step 1 and assign each student $s$ to $\mu(s)$. To show that $P^*$ is a Nash equilibrium, consider a pair of student and school $(s,c)$ such that $c P^*_s \mu(s)$. As $\mu$ is R-stable, we know that $|\mu(c)| = q_c$ and each student who is matched with school $c$ under $\mu$ is more preferred to $s$; also, for each $s' \in \mu(c)$, $\mu(c)$ is her top ranked school at $P^*$. Thus, $s$ cannot be matched to $c$ by misreporting her preferences. Since $s$ is arbitrarily chosen, the preceding argument suffices the non-existence of profitable deviations at $P^*$. We conclude that $P^*$ is a Nash equilibrium with the R-stable matching $\mu$ as its associated Nash equilibrium outcome.
\end{proof}


We now present our main argument on the asymptotic performance of the IAM with affirmative actions.

\begin{prop} \label{prop:iamnash}
The sets of Nash equilibrium outcomes of the IAM-Q and its corresponding IAM-R are outcome equivalent in large markets.
\end{prop}


Proposition \ref{prop:iamnash} implies that there is no need to distinguish these two affirmative actions under the IAM when all students are playing equilibrium strategies in large markets. Nevertheless, from both empirical and experimental evidence \citep{PS08,FN16}, we know that students are not necessarily playing their best responses under the IAM due to cognitive issues or coordination failures with other students. Such off-equilibrium behavior could be further exaggerated by the potentially large equilibria set when the market contains sufficiently many students and schools \citep{CJK18}. Thus, compared to the asymptotic outcome equivalence under the SOSM with these two affirmative actions \citep{Liu22da}, our asymptotically equivalent Nash equilibrium outcomes under the IAM could be less robust in practice, as truthful reporting is a dominant strategy for all students under the SOSM with either of these two affirmative actions \citep{AS03,HYY13}.


\section{Asymptotic Outcome Non-equivalence under the TTCM}    
\label{ssec:ttc}                          

For each market $\Gamma =(S, C, P, \succ, (q^M, r^m))$, the \emph{top trading cycles mechanism (TTCM) with affirmative actions} algorithm, which is based on the original top trading cycles algorithm proposed by \cite{SS74}, runs as follows:
\begin{itemize}

\item[] \textbf{Step 1}: Start with a matching in which no student is matched. For each school $c$, set its capacity counter at $q_c$. If $c$ has a majority quota, set its \emph{quota counter} at its majority quota $q_c^M$; if $c$ has a corresponding minority reserve, set its \emph{reserve counter} at its minority reserve $r_c^m$. If the reserve counter of school $c$ is positive, then it points to its most preferred minority student; otherwise it points to its most preferred student. Each student $s$ points to her most preferred \emph{acceptable} school that still has a seat for her, and otherwise points to herself; that is, an acceptable school $c$ whose capacity counter is strictly positive and, if $s \in S^M$, its quota counter is strictly positive. There exists at least one cycle (if a student points to herself, it is regarded as a cycle). Every student in a cycle is assigned a seat at the school she points to (if she points to herself, then she gets her outside option) and is removed. The capacity counter of each school in a cycle is reduced by one and, if: (i) the assigned student $s$ is a majority student and the school matched to $s$ has a majority quota, then reduces the quota counter of the matched school by one; (ii) the assigned student $s$ is a minority student and the school matched to $s$ has a minority reserve, then reduces the reserve counter of the matched school by one. If no student remains, terminate. Otherwise, proceed to the next step.


\item[] \textbf{Step l}: Start with the matching and counter profile reached at the end of Step $l-1$. For each remaining school $c$, if its reserve counter is positive, then $c$ points to its most preferred minority student among all remaining minority students; otherwise it points to its most preferred student among all remaining students. Each remaining student $s$ points to her most preferred \emph{acceptable} school that still has a seat for her, and otherwise points to herself; that is, an acceptable school $c$ whose capacity counter is strictly positive and, if $s \in S^M$, its quota counter is strictly positive. There exists at least one cycle (if a student points to herself, it is regarded as a cycle). Every student in a cycle is assigned a seat at the school she points to (if she points to herself, then she gets her outside option) and is removed. The capacity counter of each school in a cycle is reduced by one and, if: (i) the assigned student $s$ is a majority student and the school matched to $s$ has a majority quota, then reduces the quota counter of the matched school by one; (ii) the assigned student $s$ is a minority student and the school matched to $s$ has a minority reserve, then reduces the reserve counter of the matched school by one. If no student remains, terminate. Otherwise, proceed to the next step.

\end{itemize}

The algorithm terminates in a finite number of steps since there is at least one student matched and removed in any step of the algorithm. For a market $\Gamma$, if $r_c^m = 0$, $\forall c \in C$, i.e., a market with only majority quota, then the above algorithm reduces to the \emph{top trading cycles mechanism with majority quota} (TTCM-Q, henceforth) proposed by \cite{AS03}. Accordingly, if $q_c^M = q_c$, $\forall c \in C$, i.e., a market with only minority reserve, then the above algorithm reduces to the \emph{top trading cycles mechanism with minority reserve} (TTCM-R, henceforth) proposed by \cite{HYY13}. 


The following proposition presents our main argument on the asymptotic performance of the TTCM with affirmative actions. It implies that the majority quota and its corresponding minority reserve generate different matching outcomes with non-negligible probability under the TTCM, even in arbitrarily large markets with sufficiently many schools and a relatively slow growth of reserved seats. 

\begin{prop} \label{prop:ttc}

The TTCM-Q and its corresponding TTCM-R are not outcome equivalent in large markets.

\end{prop}

\begin{proof}
See Appendix \ref{app:ttc}.
\end{proof}



Such distinct asymptotic performance of the TTCM compared to its IAM counterparts essentially comes from the priority trade nature of the TTCM. As illustrated in the proof in Appendix \ref{app:ttc}, blocking the possible priority trades under the TTCM with affirmative actions requires that it is very unlikely for any two different students to list the same school \emph{without} reserved seats in a sequence of random markets of arbitrary sizes (i.e., school $c_1$ in Event 1 when $n \ge 4$, or school $c_2$ in the $2 \le n< 4$ case). This cannot be satisfied even under our relatively restricted regularity conditions of Definition \ref{def:regular}. By contrast, the outcome convergence process of these two affirmative actions under the IAM (of Proposition \ref{prop:iamnash}) only demands that no two different students (either majority or minority) will list the same school \emph{with} nonzero reserved seats with a high probability.


Also, since the purpose of imposing affirmative actions in school choice markets is to improve the matching outcomes (i.e., welfare) of minority students, the asymptotically non-equivalent TTCM-Q and its corresponding TTCM-R clearly induce an ambiguous Pareto dominance relationship for minorities in large markets.\ft{To see this, recall the two minority students $s_1$ and $s_2$ in the $n \ge 4$ case of the proof of Proposition \ref{prop:ttc} in Appendix \ref{app:ttc}: the welfare improvement of $s_1$ (when $s_1$ is assigned to $c_1$, with conditional probability $\pi_1 >0$) under the the TTCM-Q is clearly at the expense of $s_2$ (when $s_2$ is assigned to $c_1$, with conditional probability $\pi_2 >0$) under the TTCM-R, and vice versa.} We thus have the following corollary of Proposition \ref{prop:ttc}.

\begin{cor} \label{cor:welfare}

There exists no Pareto dominance relationship for minorities between the TTCM-Q and its corresponding TTCM-R in the sequence of random markets $(\tilde{\Gamma}^1, \tilde{\Gamma}^2,\dots)$.

\end{cor}

\section{Conclusion} 
\label{sec:conclusion}

This paper studies the asymptotic performance of two celebrated matching mechanisms, the immediate acceptance mechanism (IAM) and the top trading cycles mechanism (TTCM), in the context of school choice with affirmative actions. Different from most extant studies on large matching markets, we make a clear distinction on the asymptotic performance of the IAM and the TTCM with affirmative actions. Given the substantial political, administrative and cognitive costs in the selection process of affirmative action policies, our results provide guidance to policymakers regarding the cost-effectiveness of the IAM over its TTCM counterpart in large school choice markets with affirmative actions. Last, since we can treat affirmative actions as a generic type-specific constraint which is not limited in the context of school choice, future research can also work on identifying the asymptotic performance of conventional matching mechanisms in other markets with a large number of participants and type-specific constraints (e.g., elite college admissions, Covid-19 vaccine allocations, refugee resettlement, among others), and designing new mechanisms to improve the resource allocation effectiveness in these matching markets.

\section*{Acknowledgements}

This research is supported by the Swiss National Science Foundation (No. 100018$\_$192583), and the National Natural Science Foundation of China (No. 72173072). 

\newpage

\appendix           
\addappheadtotoc

\section{Appendix}
\label{app}

\subsection{Examples} 
\label{app:example}

\begin{ex} \label{ex1} \normalfont
\emph{(Students' different misreporting strategies under the IAM-Q and the IAM-R.)} Consider the following market $\Gamma= (S, C, P, \succ, (q^M, r^m))$ with three schools $C= \{c_1,c_2,c_3\}$, and five students $S = \{s_1, s_2, s_3,s_4.s_5\}$, where $S^M = \{s_1,s_2,s_3\}$ and $S^m = \{s_4,s_5\}$. $q_{c_1} = q_{c_3} = 1$, and $q_{c_2} = 3$. Schools and students have the following priority and preference orders:
\begin{table}[!ht]
\centering
\begin{tabular}{C{1.5cm} C{1.5cm}| C{1.5cm} C{1.5cm} | C{1cm} C{1cm} | C{1cm}}
$\succ_{c_k,\, k =1,2}$ & $\succ_{c_3}$ & $P_{s_i, \; i =1,2,3,4}$ & $P_{s_5}$ & $P'_{s_2}$ & $P'_{s_5}$ & $P''_{s_3}$\\ [0.5ex] 
\hline
$s_1$ & $s_4$ & $c_2$  & $c_1$  & $c_1$ &  $c_3$ & $c_3$\\
$s_2$ & $s_5$ & $c_1$  & $c_3$  &       &        &       \\
$s_3$ & $s_1$ & $c_3$  & $c_2$  &       &        &      \\
$s_4$ & $s_2$ &        &        &       &        &      \\
$s_5$ & $s_3$ &        &        &       &        &      \\ [0.5ex]
\end{tabular}
\end{table}

Suppose that $\Gamma$ has the following majority quota and its corresponding minority reserve: $(q_{c_1}^M, q_{c_2}^M, q_{c_3}^M) = (1,1,1)$, or correspondingly, $(r_{c_1}^m,r_{c_2}^m,r_{c_3}^m) =(0,2,0)$. When all students report their truthful preferences $P_{s_i}$, $i =1,\dots,5$, the IAM-Q and its IAM-R counterpart produce different matching outcomes as:
\begin{equation*}
f^{IAM-Q} (\Gamma) = \left(
\begin{array}{ccc}
c_1 & c_2 & c_3 \\
s_5 & \{s_1, s_4\} & s_2
\end{array} \right)\qquad\qquad
f^{IAM-R} (\Gamma) = \left(
\begin{array}{ccc}
c_1 & c_2 & c_3 \\
s_5 & \{s_1, s_2, s_4\} & s_3
\end{array} \right)
\end{equation*}
which leave $s_3$ unmatched under the IAM-Q.

Different from the strategy-proof TTCM-Q and TTCM-R \citep{AS03,HYY13}, students can benefit from manipulating their reported preferences under the IAM with either of these two affirmative actions. Consider the profitable deviations of the majority student $s_2$ and the minority student $s_5$ by reporting $P'_{s_2}$ and $P'_{s_5}$ under the IAM-Q. Given $\Gamma' = (S, C, P', \succ, (q^M, r^m))$, with $P' = (P_{s_1},P'_{s_2},P_{s_3},P_{s_4},P'_{s_5})$, the IAM-Q produces 
\begin{equation*}
f^{IAM-Q} (\Gamma') = \left(
\begin{array}{ccc}
c_1 & c_2 & c_3 \\
s_2 & \{s_1, s_4\} & s_5
\end{array} \right)
\end{equation*}
which unilaterally improves the matching outcome of the majority student $s_2$ compared to her matching outcome with the sincere preference order $P_{s_2}$. Note that although $s_5$ is strictly worse-off under $\Gamma'$ compared to her matching outcome under $\Gamma$, she will bear further welfare loss if she insists on reporting her truthful preference $P_{s_5}$. To see this, let the unmatched majority student $s_3$ submit $P''_{s_3}$, while $s_2$ still submits her strategic preference order $P'_{s_2}$. Denote $\hat{\Gamma} = (S, C, \hat{P}, \succ, (q^M, r^m))$, with $\hat{P} = (P_{s_1},P'_{s_2},P''_{s_3},P_{s_4},P_{s_5})$, the IAM-Q thus produces 
\begin{equation*}
f^{IAM-Q} (\hat{\Gamma}) = \left(
\begin{array}{ccc}
c_1 & c_2 & c_3 \\
s_2 & \{s_1, s_4, s_5\} & s_3
\end{array} \right)
\end{equation*}
in other words, $s_5$ will be hurt from not behaving strategically under the IAM-Q. Also, note that no student has further profitable deviations from reporting $P'$.

Accordingly, the unique equilibrium matching outcome under the IAM-R is
\begin{equation*}
f^{IAM-R} (\Gamma'') = \left(
\begin{array}{ccc}
c_1 & c_2 & c_3 \\
s_3 & \{s_1, s_2, s_4\} & s_5
\end{array} \right)
\end{equation*}
where $\Gamma'' = (S, C, P'', \succ, (q^M, r^m))$, with $P'' = (P_{s_1},P_{s_2},P''_{s_3},P_{s_4},P_{s_5})$; i.e., only $s_3$ has a profitable deviation from reporting her truthful preference order $P_{s_3}$ under the IAM-R.

\end{ex} 

\subsection{Proof of Proposition \ref{prop:ttc}} 
\label{app:ttc}

Consider a sequence of random markets $(\tilde{\Gamma}^1, \tilde{\Gamma}^2,\dots)$, where there are $n$ schools and $\lambda n$ students,  $\lambda \ge 1$, in each random market $\tilde{\Gamma}^n$. Assume that the preferences of all students are generated according to the preference generation procedure defined in Section \ref{ssec:large}, with uniform distribution over all schools and preference length $k = 1$. Also, assume that school priorities are drawn identically and independently from the uniform distribution over students such that all students are acceptable. For each random market $\tilde{\Gamma}^n$, denote $t_n \in (0,1)$ the portion of minority students, while $1 - t_n$ the corresponding portion of majority students. Also, assume that $q_c = 1$ or $2$ for every school $c$ in $\tilde{\Gamma}^n$, denote $\delta_n \in (0,1)$ the portion of schools with two seats, while $1 - \delta_n$ the corresponding portion of schools with one seat. The preceding assumptions guarantee that the regularity conditions of Definition \ref{def:regular} are satisfied.

Let $p^n$ be the probability that the two affirmative actions produce different outcomes under the TTCM in market $\tilde{\Gamma}^n$. We will construct examples to show that the probability $p^n$ is strictly bounded away from zero in a sequence of random markets of different sizes $(\tilde{\Gamma}^1, \tilde{\Gamma}^2,\dots)$.

$p^1 >0$ is trivially satisfied when $\tilde{\Gamma}^1$ contains one majority student $s_1$ and one minority student $s_2$, while the exact school $c_1$ has one seat, $\delta_1 \in (0,1)$, and $s_1 \succ_{c_1} s_2$. For $2 \le n < 4$, it is a positive probability event that apart from other participants in $\tilde{\Gamma}^n$, there are two schools $c_1$ and $c_2$, and three students $s_1, s_2 \in S^{M,n}$, $s_3 \in S^{m,n}$, with the following priority and preference orders:
\begin{table}[!ht]
\centering
\begin{tabular}{C{1cm} C{1cm} | C{1cm} C{1cm}}
$\succ_{c_1}$ & $\succ_{c_2}$ &  $P_{s_1}$ & $P_{s_i, \; i =2, 3}$ \\ [0.5ex] 
\hline
$s_2$ & $s_1$ & $c_1$ & $c_2$  \\
$s_3$ & $s_2$ &       &        \\
$s_1$ & $s_3$ &       &        \\ [0.5ex]
\end{tabular}
\end{table}

Assume $q_{c_1} =2$ and $q_{c_2} =1$, with the following majority quota and its corresponding minority reserve: $(q_{c_1}^M,q_{c_2}^M) = (1,1)$ or correspondingly, $(r_{c_1}^m,r_{c_2}^m) =(1,0)$. The TTCM-Q and its TTCM-R counterpart produce different matching outcomes as:
\begin{equation*}
f^{TTCM-Q} (\Gamma) = \left(
\begin{array}{cc}
c_1 & c_2  \\
s_1 & s_2
\end{array} \right) \qquad\qquad
f^{TTCM-R} (\Gamma) = \left(
\begin{array}{cc}
c_1 & c_2  \\
s_1 & s_3
\end{array} \right)
\end{equation*}
i.e., the majority student $s_2$ is assigned to $c_2$ when $c_1$ has the majority quota $q_{c_1}^M =1$, while the minority student $s_3$ is assigned to $c_2$ when $c_1$ has the corresponding minority reserve $r_{c_1}^m =1$. This gives $p^n >0$, for each $n \ge 2$.

For $n \ge 4$. Let $\lambda = 1$. Denote $c_1$ an arbitrary school with no affirmative actions, $q_{c_1} =1$. Let Event 1 be the event that there are exactly two minority students, denoted by $s_1$ and $s_2$ respectively, rank $c_1$ first, $s_1, s_2 \in S^{m,n}$. The probability of Event 1 is
\[ 
\binom{ n \, t_n}{1} \times \binom{n \, t_n-1}{1} \times \frac{1}{n^2} \times \left(1- \frac{1}{n} \right)^{n-2},
\]
where $t_n \in (0,1)$ for any arbitrarily large $n \ge 4$. We can derive its limit when $n$ approaches $\infty$ as
\begin{align*}
\lim_{n \to \infty} \; \frac{n t_n (n t_n-1)}{n^2} \left(1- \frac{1}{n}\right)^{n-2} = & \lim_{n \to \infty} \; (t_n)^2 \times \left(1- \frac{1}{n}\right)^n \times \left(1- \frac{1}{n}\right)^{-2}\\ = & \; (t_n)^2 \times \frac{1}{e} \times 1 = \frac{(t_n)^2}{e}.
\end{align*}
Thus, for any sufficiently large $n$, the probability of Event 1 is at least, say, $\frac{(t_n)^2}{2e} > 0$.

Given Event 1, consider Event 2 such that except school $c_1$, there is exactly one school (denoted by $c_2$), $q_{c_2} =2$ with either a majority quota $q_{c_2}^M =1$ or its corresponding minority reserve policy $r_{c_2}^m=1$, lists $s_1$ over all the rest students in its priority order; also, there is exactly one school (denoted by $c_3$) lists $s_2$ first. The conditional probability of Event 2 is given by
\[
\binom{n \, \delta_n}{1} \times \binom{n-2}{1} \times \frac{1}{n^2} \times \left(1- \frac{2}{n} \right)^{n-3},
\]
where $\delta_n \in (0,1)$ for any arbitrarily large $n \ge 4$. The limit of the above expression is
\begin{align*}
\lim_{n \to \infty} \; \frac{n \delta_n (n-2)}{n^2} \left(1- \frac{2}{n}\right)^{n-3} = & \lim_{n \to \infty} \; \delta_n \times \left(1- \frac{2}{n}\right)^n \times \left(1- \frac{2}{n}\right)^{-3}\\ = & \; \delta_n \times \frac{1}{e^2} \times 1 = \frac{\delta_n}{e^2},
\end{align*}
as $n$ approaches $\infty$. Thus, for any sufficiently large $n$, the conditional probability of Event 2 given Event 1 is at least, say, $\frac{\delta_n}{2e^2} >0$.

Given Event 1 and 2, consider Event 3 such that except the two minority students $s_1$ and $s_2$, there is exactly one student (denoted by $s_3$) ranks $c_2$ first and exactly one student (denoted by $s_4$) ranks $c_3$ first, where $s_3, s_4 \in S^n$. The conditional probability of Event 3 is
\[
\binom{n-2}{1} \times \binom{n-3}{1} \times \frac{1}{(n-1)^2} \times \left(1- \frac{2}{n-1} \right)^{n-4}.
\]

Similarly, we can derive the limit of this expression as $\frac{1}{e^2}$, when $n \to \infty$. For any sufficiently large $n$, the conditional probability of Event 3 given Event 1 and 2 is at least, say, $\frac{1}{2e^2}$. 

Given Events 1, 2, and 3, let Event 4 be the event that apart from other students in $\tilde{\Gamma}^n$, $c_1$ ranks $s_3$ and $s_4$ higher than both $s_1$ and $s_2$. Since Events 1-3 do not impose any restrictions on the rankings of these four students in $c_1$'s priority order, the conditional probability of Event 4 is $\frac{1}{6}$. Note that given Events 1-4 and the assumption that $k = 1$ and $q_{c_1} =1$, the event that school $c_1$ is matched with $s_1$ or $s_2$ (under either the majority quota or its corresponding minority reserve) while being contained in a cycle involving participants other than $c_1$, $s_1$, and $s_2$ occurs with conditional probability 1. 

Given Events 1-4, let $\pi_1 >0$ be the conditional probability that school $c_1$ is matched with $s_1$ when $c_2$ has the majority quota $q_{c_2}^M =1$. The unconditional probability that $c_1$ is matched with $s_1$ when $q_{c_2}^M =1$, is thus at least $\frac{\pi_1 (t_n)^2 \delta_n}{48e^5} > 0$. Accordingly, let $\pi_2 >0$ be the conditional probability (given events 1-4) that school $c_1$ is matched with $s_2$ when $c_2$ has the corresponding minority reserve $r_{c_2}^m =1$. The unconditional probability that $c_1$ is matched with $s_2$ when $r_{c_2}^m =1$, is thus at least $\frac{\pi_2 (t_n)^2 \delta_n}{48e^5} > 0$.

Therefore, for any sufficiently large $n$, we cannot eliminate the probability that these two affirmative actions generate different matching outcomes under the TTCM in market $\tilde{\Gamma}^n$; i.e., there is an $\tilde{n}$ such that $p^n > 0$, for any $n \ge \tilde{n}$. This completes the proof.

\bigskip
\bibliographystyle{ecta} 


\end{document}